\documentclass[article]{ifcolog-c}
\usepackage{lineno,hyperref}
\modulolinenumbers[5]

\newcommand{\putaway}[1]{}

\usepackage{amsthm}

\usepackage{amsmath,amssymb,tikz,bbm,mathrsfs}
\usetikzlibrary{arrows,decorations.pathmorphing,backgrounds,positioning,fit,petri}

\newtheorem{theorem}{Theorem}[section]
\newtheorem{lemma}[theorem]{Lemma}

\newtheorem{definition}[theorem]{Definition}
\newtheorem{exam}[theorem]{Example}

\newtheorem{remark}[theorem]{Remark}

\usepackage{nicefrac,amsmath,amssymb}
\usepackage{amsfonts}
\usepackage{bm}
\usepackage[inline]{enumitem}
\usepackage{stmaryrd,multicol,tikz,float}

\newcommand{\rel}{\sqsubset}

\newcommand{\simrel}{\sim}
\newcommand{\rrel}{\sqsubseteq}

\def\lb{\left\llbracket}
\def\rb{\right\rrbracket}
\newcommand{\val}[1]{\lb #1 \rb}

\newcommand{\nx}{\mathsf X}
\newcommand{\y}{\mathsf Y}

\newcommand{\eventuallyF}{\mathsf F}
\newcommand{\eventuallyP}{\mathsf P}

\newcommand{\M}{\ensuremath{\mathcal{M}}}
\newcommand{\tuple}[1]{\ensuremath{\langle #1\rangle}}
\usepackage{tikz}
\usetikzlibrary{arrows.meta,shapes,positioning,calc,automata,quotes,patterns}
\usepackage{nicefrac}

\begin{document}

\title{Common Knowledge Always, Forever}

\titlerunning{Common Knowledge Always, Forever}

\titlethanks{
This work was partially supported by the project PID2023-149556NB-I00 (Spanish Ministry of Scence and Innovation) and the project CTASP (\'etoiles montantes Pays de la Loire region, France). 
}

\addauthor[martin.dieguezlodeiro@univ-angers.fr]{Martín Diéguez}
{University of Angers}
\addauthor[fernandez-duque@ub.edu]{David F\'ernandez-Duque}
{Department of Philosophy, University of Barcelona}
\authorrunning{M. Diéguez and D. Fernández-Duque}





\maketitle

\begin{abstract}
There has been an increasing interest in topological semantics for epistemic logic, which has been shown to be useful for, e.g., modelling evidence, degrees of belief, and self-reference.
We introduce a polytopological PDL capable of expressing common knowledge and various generalizations and show it has the finite model property over closure spaces but not over Cantor derivative spaces.
The latter is shown by embedding a version of linear temporal logic with `past', which does not have the finite model property.
\end{abstract}



\section{Introduction}

It is no overstatement to say that Andreas Herzig has been a profound influence on the careers of many young researchers, including the authors during their time spent at Toulouse. He has been a mentor in the truest sense of the word, offering not only guidance on technical matters but also valuable insights on life and career choices.
Over the two years we spent there, we had the privilege of engaging in insightful discussions on a wide range of scientific topics and forging a lasting friendship. 

It is thus an honour to contribute to this volume on the occasion of his 65th birthday, and only natural to combine common knowledge~\cite{HerzigCommon} and propositional dynamic logic~\cite{HerzigPDL}, two topics close to Andreas' interests, with recent developments of the authors and collaborators in topological interpretations of dynamic logics~\cite{BoudouDF22}.

Topological semantics is in fact older than Kripke semantics and can be traced back to McKinsey and Tarski \cite{MT44}.
When the modal $\Diamond$ is interpreted as topological closure and the modal $\Box$ as topological interior, one obtains a semantics for the modal logic $\sf S4$ and its extensions, generalizing Kripke semantics over transitive, reflexive frames.
The logic of all topological spaces in this semantics is $\sf S4$; see~\cite{vBB07} for a nice overview.

McKinsey and Tarski also suggested a second topological semantics, obtained by interpreting the modal $\Diamond$ as Cantor derivative; recall that the derivative $d(A)$ of a set $A$ consists of all limit  points of $A$.
Esakia \cite{Esakia01, Esakia04} showed that the derivative logic of all topological spaces is the modal logic $\mathsf{wK4} = \mathsf{K} + (\Diamond \Diamond p \to \Diamond p\vee p$).  This is also the modal logic of all {\em weakly transitive} frames, i.e.~those for which the reflexive closure of the accessibility relation is transitive.
It is well-known that the modal logic of transitive frames is $\mathsf{K4}$ \cite{BdRV01, CZ97}, which moreover corresponds to a natural class of topological spaces denoted $T_D$.

It has been argued that topology is a natural framework for modelling epistemic attitudes~\cite{SurpriseKR,balbezozgsme17,BaltagGS16,BaltagBOS16}.
In this paper, we propose a multi-agent extension of this framework with $\sf PDL$-style operators allowing for the formation of cooperative knowledge operations including, but not limited to, common knowledge.
We show that under the better-understood variants of this framework, the resulting logics are decidable, but things can get tricky when combining more than one agent with Cantor derivative.

\paragraph{Layout} In Section~\ref{DerivSpace} we introduce the definition of derivative spaces.
In Section~\ref{sec:syntax-and-semantics} we present the syntax and the semantics of topological $\sf PDL$.
In Section~\ref{SecDec} we show that the validity problem in topological $\sf PDL$ is decidable 
when considering two specific classes of derivative spaces: the class of topological closure and the class of topological derivative spaces.
In Section~\ref{sec:embedding} we prove that topological $\sf PDL$ does not enjoy the finite model property (FMP) in the general case.
To do that, we show that we can embed $\sf LTL$ with past (which does not posses the FMP) into topological $\sf PDL$.
We finish the paper with some concluding remarks.

\section{Derivative spaces}\label{DerivSpace}

Although our primary focus in this paper is $\sf PDL$ on topological spaces, for technical reasons it is useful to consider a slightly more general class of structures which will allow us to unify topological and Kripke semantics.

\begin{definition}\label{derDS}
A \emph{derivative space} is a pair $( X ,d)$, where $X$ is a set of `points', and $d:2^X\to 2^X$ is an operator on subsets of $X$, satisfying the following properties, for all $A,B\subseteq X$:
\begin{itemize}
\item $d(\varnothing)=\varnothing$;
\item $d(A\cup B)=d(A)\cup d(B)$;
\item $d(d(A))\subseteq A\cup d(A)$.
\end{itemize}
The conjunction of the first two conditions above is known as \emph{normality}, while the third condition is \emph{weak idempotence}.
\end{definition}

The notion of a derivative space is the concrete set-theoretic instantiation of the more abstract concept of a \emph{derivative algebra}, introduced by Esakia \cite{Esakia04} (as a generalization of a notion with the same name introduced by McKinsey and Tarski \cite{MT44}).

The primary examples of derivative spaces come from topological spaces, i.e.~pairs $(X,\mathcal T)$, where $\mathcal T$ is a collection of subsets of $X$ called {\em open sets} such that $\varnothing,X\in\mathcal T$ and $\mathcal T$ is closed under arbitrary unions and finite intersections.
If $x\in U\in\mathcal T$, we say that $U$ is a {\em neighbourhood of $x$,} and we denote the set of neighbourhoods of $x$ by $\mathcal N(x)$. 
The complements of open sets are called {\em closed sets.}
Often the topology $\mathcal T$ is presented in terms of a {\em base,} which is a collection of open sets $\mathcal B$ such that whenever $U$ is a neighbourhood of $x$, there is $B\in \mathcal B$ with $x\in B\subseteq U$.

\begin{exam}[topological closure spaces]
A special case of derivative spaces is given by \emph{closure spaces}:
these are derivative spaces $(X, c)$ that additionally satisfy $A\subseteq c(A)$ (and, {\em a fortiori,} $c(c(A))\subseteq c(A)$), for every $A\subseteq X$.
These strengthened conditions are known as the {\em Kuratowski axioms,} which define topological spaces in terms of their closure operator.

Given a closure space, let  $A\subseteq X$ be {\em closed} whenever $A=c(A)$, and {\em open} whenever its complement is closed.
This gives us a topological space as defined above.
Conversely, every topology $\mathcal T$ on $X$ gives rise to an operator $c_\mathcal T$ where $c_\mathcal T(A)$ is the set of $x\in X$ so that, for every neighbourhood $U$ of $x$, $A\cap U\neq\varnothing$.
These two operations are inverses of each other, so closure spaces are exactly the same notion as topological spaces.

When considered as a special case of derivative spaces, with $d(A):=c(A)$ given by topological closure, topological spaces will be called \emph{topological closure spaces}.
\end{exam}

\begin{exam}[topological derivative spaces]
The prototypical example of derivative spaces are structures
$(X,d)$, based on an underlying topological space $(X,\mathcal T)$, but with the derivative operator given by the so-called \emph{Cantor derivative}, i.e.~by taking $d(A)$ to be the \emph{set of limit points} of $A$:
\begin{align*}
d_\mathcal T(A) & :=\, \{y\in X: y\in c_\mathcal T(A-\{y\})\} \\
&= \{y\in X:\forall U\in \mathcal{N}(y)\ A\cap (U-\{y\})\not=\emptyset\}.
\end{align*}
It is not hard to show that  $(X,d_\mathcal T)$ is a derivative space, which we will refer to as a \emph{topological derivative space}.
The closure operator can be recovered as $c_\mathcal T(A) = A\cup d_\mathcal T(A)$.
\end{exam}

When clear from context, we may omit the subindex $\mathcal T$ and write $c$, $d$ for $c_\mathcal T$, $d_\mathcal T$.
However, not every derivative space is of this form; another large family of examples comes from Kripke frames.

\begin{exam}[weakly transitive Kripke frames] A \emph{weakly transitive frame} (or {\em $\sf wK4$ frame}) is a Kripke structure $(W,\rel)$, consisting of a set of `states' (or `possible worlds') $W$, together with a binary relation ${\rel}\subseteq W\times W$ (known as an `accessibility' or `transition' relation), assumed to be \emph{weakly transitive}: i.e., for all states $w,s,t\in W$, if $w \rel s\rel t$ then $w \rrel t$.	
Given $A \subseteq W$, we set $A{\downarrow}:= \lbrace w\in W: \exists s \, w\rel s \in A\rbrace$ and $A{\Downarrow}:=  \lbrace w\in W: \exists s \, w\rrel s \in A\rbrace$. 

We write $w\simrel v $ if $w\rrel v\rrel w$.
It is easy to see that every weakly transitive frame gives rise to a derivative space $(W,d_{\rel})$, obtained by taking the derivative $d_{\rel}$ to be usual modal `diamond' operator:
\[d_{\rel}(A) := A{\downarrow} = \{w\in W: \exists s \, w\rel s \in A\}.
\]
Moreover, the induced closure $c_{\rel}(A)$ (as defined above in arbitrary de\-ri\-va\-tive spaces) is given by
$c_{\rel}(A)\, = \, A\cup d_\rel(A).$
\end{exam}

In general, weakly transitive frames are \emph{not} topological derivative spaces.
But the intersection of the two classes is of independent interest, as shown by the next two examples:

\begin{exam}[Alexandroff closure spaces as ${\sf S4}$ Kripke frames]\label{ex: S4}
A topological space $\mathcal{X}=(X,\mathcal T)$ is \emph{Alexandroff} if its closure operator distributes over arbitrary unions: $c(\bigcup_i A_i)=\bigcup_i c(A_i)$.
Given $x,y\in X$, define $x\rel y$ if $x\in c\{y\}$.
Then, $\rel$ is reflexive and transitive and one can check that if $\mathcal X$ is Alexandroff, the relational derivative coincides in this case with the topological closure: $d_{\rel}=c$.
As it is well known, the converse also holds: every ${\sf S4}$ frame $(X,\rel)$ gives rise to an Alexandroff closure space, by putting $c_{\rel}(A):=A{\downarrow}=A{\Downarrow}$ for the closure/derivative operator.
Thus, Alexandroff closure spaces can be identified with $\sf S4$ Kripke frames.
\end{exam}

\begin{exam}[Alexandroff derivative spaces as irreflexive $\sf wK4$ frames]\label{AlexDerivative}
Another way to convert an Alexandroff space $(X,\mathcal T)$ into a relational structure is to define $x\rel y$ if $x\in d\{y\}=c\{y\} - \{y\}$, for all $x,y\in X$.
Then $\rel$ is weakly transitive and irreflexive, and the relational derivative $d_{\rel}$ coincides in this case with the Cantor derivative induced by $\mathcal T$.
Conversely, every irreflexive $\sf wK4$ frame $(X,\rel)$ gives rise to an Alexandroff derivative space $(X,d_\rel)$.
So, Alexandroff topological derivative spaces are essentially the same as irreflexive $\sf wK4$ frames.
\end{exam}

Not every weakly transitive frame is irreflexive, but it is well known that every weakly transitive frame is a p-morphic image of an irreflexive weakly transitive frame which has at most twice as many points (see e.g.~\cite{TopoMuJACM}).

Here we may ask if $\sf S5$, the traditional foundation for epistemic logic, may also be regarded as a topological logic.
In fact this is the case, although we must consider a restricted class of spaces.

\begin{exam}[Monadic spaces]
If $x$ is a point in an Alexandroff space then $x$ has a least neighbourhood $U$ (the intersection of all of its neighbourhoods), but $U$ itself may contain smaller non-empty open sets, albeit not containing $x$.
If this is not the case, i.e.~if $U\neq\varnothing$ and there is no open set $U'$ such that $\varnothing \subsetneq U'\subsetneq U$, we say that $U$ is an {\em atomic open set.}
Most topological spaces (including many Alexandroff spaces) do not contain atomic opens, but {\bf monadic spaces} are a special case of Alexandroff spaces that have a base consisting of atomic opens.
Alternately, these may be characterised by the property that every open set is closed.
However, they are more familiar via their relational representation, in which $\sqsubset$ is an equivalence relation and thus we may instead write it $\sim$; in other words, monadic spaces are basically $\sf S5$ frames.
As with other spaces, these may be equipped with their closure or Cantor derivative operator, yielding two derivative spaces.
\end{exam}

\section{Syntax and Semantics}\label{sec:syntax-and-semantics}

Our goal is to consider a multi-agent extension of topological modal logics.
In order to capture notions such as common knowledge, we employ a $\sf PDL$-like syntax, given by the following language $\mathcal L^*$:
\begin{eqnarray*}
\varphi,\psi &:= & p \ \ \mid \ \ \neg\varphi \ \ \mid \ \  \varphi\wedge\psi \ \ \mid \ \ \langle \alpha\rangle \varphi\\
\alpha,\beta &:=&  a \ \ \mid \ \ \alpha;\beta  \ \mid \ \  \alpha \cup \beta \ \ \mid \ \  \alpha^*
\end{eqnarray*}

\noindent Here, $p$ ranges over atomic propositions from a set $\mathbb P$ and $a$ over `agents' in a set $\mathbb A$.
The operation $\alpha ; \beta$ is composition, $\alpha\cup \beta$ is union and $\cdot^\ast$ is iteration.
Note that we only think of atomic programs as agents; complex programs may be viewed as communication tasks or aggregated notions of knowledge.
The classic example of this is common knowledge, which may be defined by $C_{a_1\ldots a_n}\varphi=[(a_1\cup \ldots\cup a_n)^*]\varphi$.

For semantics, each agent $a$ is endowed a derivative operator $d_a$ over some set $X$.
A structure $(X,(d_a)_{a\in\mathbb A})$ is a {\em polyderivative space,} and we inherit terminology from the single-agent setting, so that e.g.~a {\em polytopological closure space} is one where each $d_a$ is the closure operator of a topology $\mathcal T_a$ on $X$ (in which case, we may write $c_a$ instead of $d_a$).
If the operator $d_a$ is either the closure or derivative of some topology $\mathcal T_a$, we refer to elements of the latter as {\em $a$-open} sets.

We then assign an operator $\val \alpha \colon 2^X\to 2^X$ to each program inductively by
\begin{itemize}

\item $\val a (Y)= d_a(Y)$

\item $\val {\alpha;\beta}(Y) = \val\alpha  ( \val \beta (Y))$

\item $\val{\alpha \cup \beta}(Y) = \val \alpha (Y) \cup \val \beta (Y)$

\item $\val{\alpha^*}(Y) = \bigcap \{Z : \val\alpha (Z) \cup Y \subseteq Z  \} $.

\end{itemize}
\noindent 
A {\em polyderivative model} is a structure $\mathcal M = (X,(d_a)_{a\in\mathbb A},\val\cdot)$, consisting of a polyderivative space equipped with a valuation $\val\cdot\colon \mathbb P \to 2^X$.
The truth set $\val\varphi$ of a formula $\varphi$ is then defined in the usual way, with $\val{\langle\alpha\rangle\varphi} = \val{\alpha}(\val\varphi)$.
As usual, we write $(\mathcal M,x )\models\varphi$ if $x\in\val\varphi$.

Note that $\val{\alpha^*}(Y)$ is the least superset of $Y$ which is a fixed point under $\val\alpha $.
In relational structures this will yield the usual transitive closure, but these operators may be more complex in a topological setting.
In particular, $x\in \val{C_{a_1\ldots a_n}\varphi} $ iff there is $U\subseteq \val\varphi$ which is open in {\em all} topologies $ \mathcal T_{a_i}$ with $x\in U$~\cite{SarenacCommon}.

\section{Decidability Results}\label{SecDec}

As long as we restrict our attention to `well-behaved' derivative spaces, our topological $\sf PDL$ can be embedded into standard $\sf PDL$ and hence inherits some of its nice properties.
Recall that a logic has the {\em effective finite model property} if there is a computable function $f$ such that any satisfiable formula $\varphi$ can be satisfied in a model of size at most $f(|\varphi|)$, where $|\varphi|$ is the length of $\varphi$; as long as the model-checking problem is decidable, this yields the decidability of the validity problem, as it suffices to search for a model of $\neg\varphi$ of size at most $f(|\neg\varphi|)$.

\begin{theorem}
$\mathcal L^*$ has the effective finite model property over the class of closure spaces and hence the validity problem for this class is decidable.
\end{theorem}

\begin{proof}[Proof sketch]
We use a well-known trick of embedding $\sf S4$ in $\sf PDL$ (see e.g.~\cite{GH17}).
Given a formula $\varphi$, define $\varphi^*$ by replacing every instance of an {\em atomic} program $a$ by $a^*$.
For example, if  $\varphi=\langle (a\cup b)^*\rangle q$ then $\varphi^* = \langle (a^*\cup b^*)^*\rangle q$.
Then, for a Kripke model $\mathcal M = (W,(R_a)_{a\in \mathbb A},\val\cdot)$ and $w\in W$, it is readily checked that $(\mathcal M,w) \models \varphi^*$ iff $(\mathcal M^*,w)\models\varphi $, where $\mathcal M^* = (W,(R_a^*)_{a\in \mathbb A},\val\cdot)$; i.e., all accessibility relations are replaced by their transitive, reflexive closures.
We moreover can check that for any polytopological closure model $\mathcal X =( X,(c_a)_{a\in \mathbb A},\val\cdot)$ and $x\in X$,  $ (\mathcal X,x)\models\varphi$ iff $(\mathcal X,x)\models \varphi^*$, basically because $c_a=c_ac_a$ holds on any closure space so also $c^*_a=c_a$.

Now, suppose that $\varphi$ is satisfied on some polytopological closure model $\mathcal X$.
Then, so is $\varphi^*$.
By inspection on any standard axiomatisation for $\sf PDL$ (see e.g.~\cite[Chapter 10]{GoldblattBook}), we see that all axioms and rules are sound for the class of polytopological spaces, hence $\varphi^*$ is consistent.
By Kripke-completeness and FMP of $\sf PDL$, $\varphi^*$ is satisfied on some finite model $\mathcal M$, so that $\varphi$ is satisfied on $\mathcal M^*$, as needed.
\end{proof}

The logic $\sf wK4$ cannot be treated in this way, since the weakly transitive closure is not definable in $\sf PDL$.
However, we can treat the Cantor derivative of $T_D$ spaces, with the caveat that we no longer obtain a true finite model property.
Recall that $T_D$ spaces are those in which $\Diamond\Diamond p\to \Diamond p$ is valid, i.e.~those validating $\sf K4$.

\begin{theorem}
The validity problem for $\mathcal L^*$ over the class of $T_D$ topological derivative spaces is decidable.
\end{theorem}

\begin{proof}[Proof sketch]
We begin as above, but instead use $\varphi^+$ which replaces $a$ by its transitive (but not reflexive) closure, $a^+$, to show that $\varphi$ is satisfiable if and only if it is satisfied on an effectively bounded $\sf K4$ multirelational Kripke model $\mathcal M=(W,(\sqsubset_a)_{a\in\mathbb A},\val\cdot)$.

Note that $\mathcal M$ may not be irreflexive and hence not correspond to a topological derivative, but by a standard tree unwinding (see e.g.~\cite{CZ97}), we may obtain a bisimilar irreflexive frame.
Namely, we obtain a new model $\mathcal M^+=(W^+,(\sqsubset^+_a)_{a\in\mathbb A},\val\cdot^+)$, where $W^+ $ is the set of non-empty sequences of elements of $W$, with $ \ell (\vec w)$ denoting the last element of $\vec w$.
We set $\vec w\sqsubset^+_a\vec v$ iff $\vec w$ is a proper initial segment of $\vec v$, and $\vec w \in \val p^+ $ iff $\ell(\vec w)\in \val p$.
Then it is easy to check that $\sqsubset^+_a$ is transitive and irreflexive for all $a$, thus generating a $T_D$ topology on $W^+$, and $\ell\colon W^+\to W$ is a surjective p-morphism.
\end{proof}


\section{Embedding $\sf LTL$}\label{sec:embedding}

It is not surprising that our logic does not have the finite model property over topological spaces, since this is already the case for the single-agent $\sf K4$; even the formula $\Diamond\top\wedge\Box\Diamond \top$ has no finite $T_D$ topological derivative models.
However, we can get around this by allowing for transitive Kripke models, which do have the finite model property and give rise to the same logic.
In a sense, the situation is even simpler for the class of all topological derivative spaces: their logic, $\sf wK4$, also has the finite model property, and since every finite weakly transitive frame is bisimilar to a finite, irreflexive, weakly transitive frame, it also has the finite model property for the class of `true' topological derivative spaces.
Moreover, even the full $\mu$-calculus enjoys the finite model property over this class~\cite{TopoMuJACM}.
It is thus natural to ask if this remains true for our polytopological $\sf PDL$.

The answer, it turns out, is negative.
In order to prove that topological derivative $\sf PDL$ does not enjoy the finite model property in general, we present an embedding of a version of $\sf LTL$ with past ($\sf PLTL$), which lacks the finite model property. We first introduce the syntax and the semantics of this version of $\sf PLTL$.
Its language, $\mathcal L^{\sf PLTL}$, is given by the BNF grammar

\begin{equation*}
\varphi,\psi ::= p \mid  \neg \varphi \mid \varphi \wedge \psi \mid  \nx\varphi \mid \y \varphi \mid \eventuallyF \varphi \mid \eventuallyP \varphi,
\end{equation*}

\noindent where $p$ belongs to a set propositional variables  $\mathbb P$. Here, $\nx$ is read as `ne\textbf xt', $\y$ as `\textbf yesterday, $\eventuallyF \varphi$ as `\textbf future' and 
$\eventuallyP \varphi$ as `\textbf past'.

We are specifically interested in structures where both past and future are unbounded and linear, which forces the transition relation to be a bijective function.
A \emph{bijective frame} is thus a structure $\mathcal{F} = (X,S)$, with $S: X \to X$ a bijection. By $S^{-1}$ we will denote the \emph{inverse} function of $S$. 
In the standard way, given $w \in X$, we define $S^0(w) = w$ and, for all $k \ge 0$, $S^{k+1}(w) = S(S^k(w))$ (resp. $S^{-(k+1)}(w) = S^{-1}(S^{-k}(w))$).   
A $\sf PLTL$ model is a tuple $\M=\tuple{\mathcal{F},V}$, where $\mathcal{F}$ is a bijective frame and $\val\cdot : \mathbb P \to 2^X$ is a valuation. 
The \emph{satisfaction} of a temporal formula $\varphi$ at $w \in X$ is defined recursively below

	\noindent\begin{enumerate}[itemsep=0pt]
		\item $(\M, w) \models p $  iff $w\in\val p $;
		\item $(\M, w) \models \neg \varphi$ iff $(\M, w) \not \models \varphi$;
		\item $(\M, w) \models \varphi\wedge \psi$  iff $ (\M, w) \models \varphi $ and $(\M, w) \models \psi$;
		\item $(\M, w) \models \nx \varphi $   iff $ (\M, S(w)) \models \varphi $;
		\item $(\M, w) \models \y \varphi $   iff $ (\M, S^{-1}(w)) \models \varphi $;
		\item $(\M, w) \models \eventuallyF \varphi$ iff there exists $k \ge 0$ s.t.  $(\M, S^{k})(w) \models \varphi$;	
		\item $(\M, w) \models \eventuallyP \varphi$ iff there exists $k \ge 0$ s.t.  $(\M, S^{-k}(w)) \models \varphi$.
	\end{enumerate}
	We may also write $w\in\val\varphi$ if $(\M, w) \models\varphi$.
Bijective frames are a natural class of structures in which to interpret $\mathcal L^{\sf PLTL}$, since their logic is that of $\mathbb Z$ (studied in e.g.~\cite{degola16a}).
A bijection is needed for the logic to be invariant under swapping `past' and `future', since in this case $S$ and $S^{-1}$ are both functions.
However, $\mathcal L^{\sf PLTL}$ does not have the finite model property for this class.

\begin{lemma}
$\mathcal L^{\sf PLTL}$ does not enjoy the finite model property within the class of bijective frames.
\end{lemma}

\begin{proof} Let us consider the formula $\varphi =  \eventuallyF q   \wedge \neg \eventuallyP q $. We will show that $\varphi$ is satisfiable in an infinite (bijective) model but it is unsatisfiable for the class of finite bijective models.

For satisfaction, we consider a model $\mathcal M = (\mathbb Z,S,\val\cdot)$ where $S(x)=x+1$ and $\val q =\{1\}$.
It is then easy to see that $(\mathcal M,0)\models \eventuallyF q   \wedge \neg \eventuallyP q $.

Now, suppose that $\mathcal M=(X,S,\val\cdot) $ is a finite bijective model and fix $x_0\in X$.
Since $S$ is a bijection and $X$ is finite, there must be $n>0$ such that $S^n(x_0) =x_0 $.
If $x_0$ satisfies $\varphi$, then $(\mathcal M,x_0)\models \eventuallyF q $, which means that for some $k$, $(\mathcal M,S^k(x_0))\models \eventuallyF q $, and we may take $k<n$ since otherwise we may replace $k$ by $k-n$.
But then, $(\mathcal M,S^{k-n}(x_0))\models q $ and $k-n<0$, witnessing that $(\mathcal M, x_0)\models \eventuallyP q $: a contradiction.
\end{proof}

Let $a,b$ be agents (Alice and Bob) and ${\tt whole}$ be a designated atom.
For $\iota\in \{a,b\}$, define a formula

\begin{eqnarray*}
{\tt Two}_\iota &:= & (({\tt whole} \to ([\iota]\neg {\tt whole} \wedge \langle \iota \rangle \neg {\tt whole}))\\
				 &&   \wedge (\neg {\tt whole}\to ([\iota] {\tt whole} \wedge \langle \iota \rangle  {\tt whole}))).	
\end{eqnarray*}
\noindent We then let ${\tt Two}:= {\tt Two}_a\wedge {\tt Two}_b $.

Given formulas $\varphi,\psi \in\mathcal L^*$, suppose that $\varphi\wedge C_{ab}\psi $ is satisfied on some point $x_0 \in X$, where $\varphi$ is a formula where only the agents $a$ and $b$ may appear.
Then, by the topological characterisation of common knowledge, $\val{C_{ab}\psi}$ is both $a$-open and $b$-open, hence if we let $\mathcal M' = (X',d'_a,d'_b,\val\cdot')$ be the restriction of $\mathcal M$ to $\val{C_{ab}\psi}$, in the sense that $X'= \val{C_{ab}\psi}$, $d'_\iota(Y) = d_\iota(Y\cap X') \cap X'$ and $ \val p'=\val p\cap X'$.
We see that also $(\mathcal M',x_0)\models\varphi$ and, moreover, $\mathcal M'\models \psi$; here we are using the well-known fact that restricting models to an open subset preserves truth of formulas \cite{GabelaiaMasters}.
In other words, satisfiability of $\varphi\wedge C_{ab}\psi $ is equivalent to satisfiability of $\varphi$ over the class of models which validate $\psi$.

Our reduction will concern formulas of the form $\varphi\wedge C_{ab}\tt Two$, hence we may work over models validating $\tt Two$.
In the sequel, we fix a bitopological model $\mathcal M=(X,d_a,d_b,\val\cdot)$ such that $\mathcal M\models \tt Two$; additional agents may also be included in the model, but we will work with formulas involving only $a$ and $b$, and these are not affected by other agents' derivative operators.
However, our results do require at least two agents.
For $ \iota \in \{a,b\}$, we define $x\sim_\iota y$ if $x,y$ are not distinguishable by $T_\iota$, i.e.~if the two belong to exactly the same open sets for $\iota$.

\begin{lemma}\label{lemSequence}
For $\iota\in \{a,b\}$ and $x\in X$, there is a unique $y:=S_\iota(x)$ such that $y\neq x$ and $\{x,y\}$ is an atomic $\iota$-neighbourhood.
\end{lemma}

\begin{proof}
Fix $x\in X$ and $\iota\in \{a,b\}$.
Note that the uniqueness of such a $y$ is immediate since $\{x,y\}$ and $\{x,y'\}$ both being atomic implies $y=y'$, so it remains to prove existence.

Assume that $x\in \val{\tt whole}$; the case where $x\notin\val{\tt whole}$ is symmetric.
Then, $x$ satisfies $[\iota]\neg{\tt whole}\wedge \langle \iota \rangle \neg {\tt whole} $, hence $x$ has an $\iota $-neighbourhood $U$ where every point except $x$ satisfies $\neg \tt whole$, and moreover $U$ contains at least one point $y$ satisfying $\neg {\tt whole}$.

Since $y$ also satisfies ${\tt Two}_\iota $, $y$ must also have a neighbourhood $U'$ where every point but $y$ satisfies $\tt whole$.
Now consider $U'' = U\cap U'$; $U''$ is a neighbourhood of $y$, so at least one point satisfies $ \tt whole$, but it is a subset of $U$, so the only available such point is $x$ and thus $x \in U''$.
Since no points of $U''$ but $x$ satisfy $\tt whole$ and no points of $U''$ but $y $ satisfy $\neg\tt whole$, we must have that $U''=\{x,y\}$.
We finally observe that no non-empty proper subset of $ U''$ can be open since this would lead to $\langle \iota \rangle \tt whole$ or $\langle \iota \rangle\neg \tt whole$ failing at the respective point.
Thus $y$ has all required properties.
\end{proof}

Note that $S_\iota^2$ is the identity since the conditions are symmetric on $x$ and $y$, so $S_\iota$ is a bijection.
We thus obtain the following.

\begin{lemma}\label{lemBij}
The functions $S_a,S_b\colon X\to X$ are bijections, as is $S:= S_b\circ S_a$.
\end{lemma}

With this, we may define a translation $\cdot^{\rm top}\colon\mathcal L^{\sf PLTL}\to\mathcal L^*$ as follows:
\begin{itemize}

\item $({\sf X}\varphi)^{\rm top} = \langle a;b\rangle \varphi^{\rm top}$

\item $({\sf Y}\varphi)^{\rm top} = \langle b;a\rangle \varphi^{\rm top}$

\item $({\sf F}\varphi)^{\rm top} = \langle (a;b)^* \rangle \varphi^{\rm top}$

\item $({\sf P}\varphi)^{\rm top} = \langle (b;a)^* \rangle \varphi^{\rm top}$

\end{itemize}
and letting $\cdot^{\rm top}$ fix atoms and commute with Booleans.

\begin{lemma}\label{lemTrans}
Let $\mathcal M=(X, d_a, d_b,\val\cdot)$ be a derivative model with $\mathcal M\models \tt Two$.
Let $S:=S_b\circ S_a$ and define $\mathcal M^{\sf PLTL} := (X, S,\val\cdot)$.

Then, for any $x\in X$, $(\mathcal M,x)\models \varphi^{\rm top} $ iff $(\mathcal M^{\sf PLTL},x)\models \varphi $.
\end{lemma}

\begin{proof}
Induction on $\varphi$ checking that in each case, the semantic clauses coincide with their translation.
\end{proof}

\begin{lemma}\label{lemSat}
If $\varphi$ is satisfiable on a bijective model $\mathcal M$ then $\varphi^{\rm top}$ is satisfied on a bitopological derivative model $\mathcal M^{\rm top}$ such that $\mathcal M^{\rm top}\models\tt Two$.
Moreover, $\mathcal M^{\rm top}$ may be chosen so that both $\mathcal T_a$ and $\mathcal T_b$ are monadic.
\end{lemma}

\begin{proof}
Using standard unwinding techniques, it is not hard to see that every $\sf PLTL$ formula satisfied on a bijective model is satisfied on $0$ on a model of the form $\mathcal M = (\mathbb Z,S,\val\cdot)$, where $S$ is the successor function.
We may moreover assume that the variable $\tt whole$ is fresh and does not appear in $\sf PLTL$ formulas.

Consider a bitopological model on $\frac 12\mathbb Z $ (i.e., all multiples of $\nicefrac 12$)
where $\val {\tt whole} =\mathbb Z $ (and $\val p$ is unchanged for other variables).
Let $\mathcal T_a$ be the topology generated by the basis $\{\{n,n+\nicefrac 12\}: n\in\mathbb Z\}$ and $\mathcal T_b$ be generated by the basis $\{\{n-\nicefrac 12,n\}: n \in\mathbb Z\}$ (see Figure~\ref{fig:ex:sat}).
This space is easily checked to be monadic for both $\mathcal T_a$ and $\mathcal T_b$.
Then, set $\mathcal M^{\rm top}:=(\frac 12 \mathbb Z,d_a,d_b,\val\cdot)$.

It is readily checked that $\mathcal M^{\rm top} \models \tt Two$ and, moreover, if $n\in \mathbb Z$ then $S(n) = S_bS_a(n)$ (note that $S_bS_a (n+\nicefrac 12)=n-\nicefrac 12$, but this does not matter since such points are not in the orbit of $0$).
By Lemma~\ref{lemTrans}, we have that $(\mathcal M^{\rm top},0)\models\varphi^{\rm top}$, as desired.
\end{proof}

\begin{figure}[h!]\centering
\begin{tikzpicture}[node distance=1.1cm,minimum size=2mm]

	\node[] (a) {\tiny $\mathbf{0}$};
	\node[right of=a] (b) {\tiny $\nicefrac{1}{2}$};
	\node[right of=b] (c) {\tiny $\mathbf{1}$};
	
	\node[right of=c] (d) {\tiny $\nicefrac{3}{2}$};
	\node[right of=d] (e) {\tiny $\mathbf{2}$};
	\node[right of= e] (l) {$\cdots$};
		
	\node[left of=a] (g) {\tiny $\nicefrac{-1}{2}$};
	\node[left of=g] (h) {\tiny $\mathbf{-1}$};
	\node[left of=h] (i) {\tiny $\nicefrac{-3}{2}$};
	\node[left of=i] (j) {\tiny $\mathbf{-2}$};
	
	\node [left of= j] (m) {$\cdots$};

\path[color=red]
	(a) edge[] node [above] {\tiny $a$} (b)
	(c) edge[] node [above] {\tiny $a$} (d)
	(e) edge[] node [above] {\tiny $a$} (l)
	(h) edge[] node [above] {\tiny $a$} (g)
	(j) edge[] node [above] {\tiny $a$} (i);

\path[color=blue]
	(b) edge[] node [above] {\tiny $b$} (c)
	(d) edge[] node [above] {\tiny $b$} (e)
	(g) edge[] node [above] {\tiny $b$} (a)
	(i) edge[] node [above] {\tiny $b$} (h)
	(m) edge[] node [above] {\tiny $b$} (j);	
	
\path[->, dashed]
(a.north) edge[bend left] node [above] {\tiny $S$} (c.north)
(c.north) edge[bend left] node [above] {\tiny $S$} (e.north)
(h.north) edge[bend left] node [above] {\tiny $S$} (a.north)
(j.north) edge[bend left] node [above] {\tiny $S$} (h.north);

\path[->,dashed]
(b.south) edge[bend left] node [below] {\tiny $S$} (g.south)
(g.south) edge[bend left] node [below] {\tiny $S$} (i.south)
(i.south) edge[bend left] node [below] {\tiny $S$} (m.south)
(d.south) edge[bend left] node [below] {\tiny $S$} (b.south)
(l.south) edge[bend left] node [below] {\tiny $S$} (d.south);

\end{tikzpicture}
\caption{The bitopological representation of a bijective frame. Numbers in boldface represent the points where the variable \texttt{whole} holds.
Dashed arrows illustrate the successor function, while red and blue lines represent the equivalence relations associated to agents $a$ and $b$, respectively.}
\label{fig:ex:sat}
\end{figure}
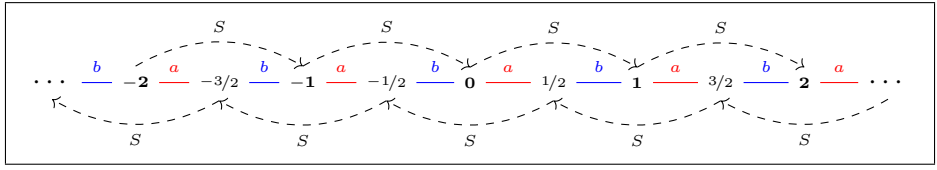

 From this, we are ready to prove that topological $\sf PDL$ does not have the finite model property.

\begin{theorem}
Suppose that $a,b $ are distinct elements of $\mathbb A$.
Let $\mathcal C$ be a class of polyderivative spaces of the form $(X,(d_a)_{a\in\mathbb A})$, where $d_a$ and $d_b$ are monadic.
Then, $\mathcal L^*$ does not have the FMP for the class of Cantor derivative spaces based on $\mathcal C$.
\end{theorem}

\begin{proof}
Just choose any $\sf LTL $ formula $\varphi$ which is satisfiable, but not finitely satisfiable; for example, the aforementioned $ \varphi :={\sf F}q\wedge\neg{\sf P}q$.
By Lemma~\ref{lemSat}, $\psi := \varphi ^{\rm top} \wedge C_{ab} \tt Two$ is satisfiable on $\mathcal C$, so it remains to check that $\psi $ is not satisfiable in any finite topological model.

But if $\mathcal M$ is a model satisfying $\psi$ then by restricting to $\val{  C_{ab} \tt Two}$, we obtain a model $\mathcal M_0$ with $\mathcal M_0\models \tt Two$ and such that $\mathcal M_0$ satisfies $\varphi ^{\rm top}$.
By Lemma~\ref{lemTrans}, $\mathcal M_0^{\sf PLTL}$ satisfies $\varphi$, hence $\mathcal M_0^{\sf PLTL}$ is infinite; however, $\mathcal M_0^{\sf PLTL}$ has the same domain as $\mathcal M_0$, which is a submodel of $\mathcal M $, so the latter must be infinite.
\end{proof}

\begin{remark}
Since monadic spaces are simply the irreflexive part of $\sf S5$ models, the lack of finite model property holds for standard $\sf S5$ multiagent epistemic models, provided the logic is enriched with an irreflexive version of each agent's knowledge modality.
\end{remark}

\section{Concluding remarks}

We have introduced a polytopological version of $\sf PDL$, and showed that some variants enjoy the finite model property but the more general version does not.
The latter is proven by embedding $\sf LTL$, from which we obtain also a {\sc pspace} complexity lower bound, although this is not too surprising since the logic already extends $\sf wK4$ which is also {\sc pspace}-complete, and $\sf S5$ epistemic logic with common knowledge is {\sc exptime}-complete~\cite{Halpern1992}.
However, the lack of FMP by no means indicates its undecidability, with bijective $\sf PLTL$ itself being an example of a decidable logic with no FMP.
We conjecture that polytopological $\sf PDL$ is indeed decidable, but leave this for future inquiry.

It is worth noting that any operator of the form $\alpha^*$ behaves like a new topological closure operator, but given that agents are endowed with derivative operators, it may be more natural to consider instead the {\em weakly} transitive closure of $\alpha$, say $\alpha^\sharp$.
This in particular would yield a `Cantor derivative' version of common knowledge via $(\alpha_1\cup\ldots\cup\alpha_n)^\sharp$, which should be more expressive than the standard one and be useful in determining when coalitions of agents can gain common knowledge of the current world.
One particular challenge here is that weakly transitive closure is not $\mu$-calculus definable, so FMP and decidability results could not be obtained using the techniques of Section~\ref{SecDec}.

Similarly, we may consider the {\em join} of the topologies, which would yield a topological notion of distributed knowledge (see e.g.~\cite{BaltagTopoMult}), itself with a `closure' and `derivative' variant.
To what extent the decidability and finite model properties (or lack thereof) extends to extensions of our logic with such operators remains an interesting open problem.

\section{Acknowledgements}

We would like to reiterate our appreciation for Andreas' support and guidance and are grateful to the anonymous referees for their kind reports.

David Fernández-Duque was supported by the Spanish Ministry of Science and Innovation grant PID2023-149556NB-I00.

\bibliography{biblio}

\end{document}